\documentclass[a4paper,UKenglish,cleveref, autoref, thm-restate,numberwithinsect]{lipics-v2021}

\hideLIPIcs  

\allowdisplaybreaks

\title{Sequencing Stochastic Jobs with a Single Sample} 

\author{Puck {te Rietmole}}{Department of Mathematics,  Utrecht University, Utrecht, The Netherlands \and
\url{https://www.linkedin.com/in/puck-te-rietmole-883692264/} }{p.terietmole@students.uu.nl}{}{}

\author{Marc Uetz}{Mathematics of Operations Research, University of Twente, Enschede, The Netherlands \and \url{https://people.utwente.nl/m.uetz/}
}
{m.uetz@utwente.nl}{https://orcid.org/0000-0003-4223-2435}{}

\authorrunning{P.\ te Rietmole and M.\ Uetz} 

\Copyright{Puck te Rietmole and Marc Uetz} 



\category{} 

\relatedversion{} 

\acknowledgements{The first author thanks Rob Bisseling for advice and support. The second author thanks Wouter Fokkema and Ruben Hoeksma for helpful discussions that initiated this work, and Jos\'e Verschae for advice.}

\nolinenumbers 

\newcommand{\alg}{\texttt{\textup{SAM}}}
\newcommand{\p}{P}
\def\Exp{{\mathbb{E}}}
\newcommand{\E}[1]{\mbox{$\Exp[\,#1\,]$}}
\newcommand{\prob}[1]{\mbox{$\mathbb{P}[\,#1\,]$}}
\newcommand{\opt}[1]{\mbox{$\texttt{WSPT}(#1)$}}
\newcommand{\eps}{\varepsilon}

\newcommand{\rand}{\texttt{\textup{RND}}}
\newcommand{\inst}{\mathcal{I}}
\newcommand{\cost}{\textup{cost}}
\newcommand{\rog}{\textup{rog}}
\newcommand{\diff}{\Delta}

\usepackage{todonotes}

\EventEditors{John Q. Open and Joan R. Access}
\EventNoEds{2}
\EventLongTitle{42nd Conference on Very Important Topics (CVIT 2016)}
\EventShortTitle{CVIT 2016}
\EventAcronym{CVIT}
\EventYear{2016}
\EventDate{December 24--27, 2016}
\EventLocation{Little Whinging, United Kingdom}
\SeriesVolume{42}
\ArticleNo{23}

\begin{document}

\maketitle

\begin{abstract}
This paper revisits the well known single machine scheduling problem to minimize total weighted completion times. The twist is that job sizes are stochastic from unknown distributions, and the scheduler has access to only a \emph{single} sample from each of the distributions. For this restricted information regime, we analyze the simplest and probably only reasonable scheduling algorithm, namely to schedule by ordering the jobs by weight over sampled processing times.  In general, this algorithm can be tricked by  adversarial input distributions, performing in expectation arbitrarily worse even in comparison to choosing a random schedule. The paper suggests notions to capture the idea that this algorithm, on reasonable inputs, should exhibit a provably good expected performance.  Specifically, we identify three natural classes of input distributions, such that for these classes, the algorithm performs better than random on any input. 
\end{abstract}

\section{Introduction, Motivation, and Model}
\label{sec:intro}

A classical result in the  literature on scheduling algorithms is due to Smith \cite{Smith56}, showing how to minimize the total weighted completion time $\sum_j w_jC_j$ of $n$ independent, non-preemptive  jobs with weights $w_j$ and processing times $p_j$ on a single machine. They should be scheduled in the so-called WSPT order, that is, jobs must be ordered by non-increasing ratios $w_j/p_j$. The same is still true when the jobs' processing times are not known in advance, but instead governed by independent random variables $\p_j$. Indeed, when minimizing 
the expected 
total weighted completion times, $\E{\sum_j w_jC_j}=\sum_j w_j\E{C_j}$, the same exchange argument shows that scheduling the jobs in order of non-increasing ratios $w_j/\E{\p_j}$ is optimal in expectation~\cite{Rothkopf66}. 
Likewise, also for more general scheduling models, specifically on more than a single machine, some of the algorithmic framework that has been developed for computing approximately optimal solutions to deterministic scheduling problems, e.g. \cite{HSSW97,CMNS2001}, can be generalized to the more general setting with stochastic processing times, e.g.~\cite{MSU99,SU2005,MUV2006,GMUX2020,GMUX2021,J2023}. However  basically all of the previously cited approximation algorithms for stochastic scheduling problems, and also Rothkopf's result~\cite{Rothkopf66} need to assume that the expected processing times $\E{\p_j}$ are known.

This paper asks the simple question what can be done if this is \emph{not} the case. The paper therefore revisits the well understood single machine scheduling problem to minimize the total weighted completion times of jobs $\sum_j w_jC_j$.
We ask how much of the optimality of Smith's result~\cite{Smith56} can be recovered if the processing times  $p_j$ that the scheduler uses to determine the schedule, is in fact a \emph{sample from an unknown distribution $\p_j$}. In that setting, one should naturally minimize 
the expected total weighted completion times of jobs $\E{\sum_j w_jC_j}$.  We can equivalently see the problem as a \emph{stochastic} single machine scheduling problem, 
but with the assumption that the scheduler has
access to only a \emph{single sample} of the processing time distributions.

This question relates to some recent work on prophet inequalities. There, it is known that a single sample per item (instead of knowing the whole distribution per item) is sufficient to get the same performance guarantee, namely $1/2$, for a simple online threshold policy, relative to the benchmark (the prophet) that knows which item will turn out to have the maximum value. This is true for the classical prophet inequality setting where one item is chosen in an online fashion~\cite{RWW2019}, and also for two items~\cite{PS2023}. One major difference with the scheduling problem considered in this paper is that we do not have a natural notion of a prophet, which means that we have to come up with a reasonable benchmark to compare to. This is discussed in  Sections~\ref{sec:prelim} and~\ref{sec:rand}.

The question can also be interpreted from the perspective of learning augmented algorithms. There, the goals is to find improved algorithms under the assumption that certain problem parameters can be learned by some hypothetical (machine) learning device. The better the learner, the better the resulting approximation guarantee, yet recovering a reasonable guarantee if the learner is bad.
As to scheduling, examples of work in this direction are \cite{KPS18,LLMV20,BMRS20,WZ20}; see \cite{alg-predictions-website} for more references. Our information regime about learning stochastic job sizes is minimalistic,  namely through one sample $p'_j$ from $\p_j$ only, for all jobs~$j$.

It is maybe no surprise that, without any further assumptions, one can define malicious input distributions that yield ``wrong'' samples with high probability, rendering an algorithm that follows these samples to have arbitrarily bad performance. Such an example is given below in Section~\ref{sec:prelim}.
Our main result is the identification of sufficient conditions on the input distributions that rule out these undesirable effects.  This is formalized using a notion of \emph{relative optimality gap}, a scaled variation on the usual notion of approximation ratio, and in comparison to an adversary that chooses a schedule uniformly at random. We also argue that our positive results are tight in a mild sense, meaning that without any of these conditions, there is a corresponding counterexample.

\section{Preliminaries: Single Machine Scheduling by Samples}
\label{sec:prelim}
We consider $n$ jobs with independent processing time  distributions $\p_j$ and weights $w_j$, $j=1,\dots, n$.  The jobs have to be scheduled non-preemptively on a single machine with the goal to minimize $\sum_j w_jC_j$, where $C_j$ denotes the completion time of job $j$.  Let $p:=(p_1,\dots,p_n)$ denote possible realizations of $P:=(\p_1.\dots\p_n)$. We denote by $I=(w,\p)$ an instance. If processing time realizations $p\sim\p$ are known, the only optimal solutions are the sequences in so-called \texttt{WSPT} order, ratios $w_j/p_j$ non-increasing~\cite{Smith56}.

If the processing times are stochastic, the solution is a stochastic scheduling policy that is non-anticipatory in the sense of \cite{MRW84}, meaning that it cannot use information about actual realizations $p_j\sim \p_j$ before scheduling a job~$j$. If $\Pi$ denotes such a scheduling policy, let $C_j^{\Pi}(p)$ be job $j$'s completion time under policy $\Pi$ for realization $p$, then the expected cost of policy $\Pi$ on instance $I=(w,\p)$ is
\[
\cost_I(\Pi):= \mathbb{E}_{p\sim P}\left[\sum\nolimits_jw_jC_j^{\Pi}(p)\right]\,.
\]
Also let $\opt{p}$ be the minimal cost for realization~$p$, achieved by a \texttt{WSPT} order.
Our goal is to minimize the \emph{expected regret of $\Pi$}, which is minimizing
\[
\mathbb{E}_{p\sim P}\left[\sum\nolimits_jw_jC_j^{\Pi}(p)-\opt{p}\right]= \cost_I(\Pi)- \mathbb{E}_{p\sim P}\left[\opt{p}\right]\,.
\]
In other words, we effectively seek a policy $\Pi$ minimizing $\cost_I(\Pi)$. Since the processing times are independent across jobs, it is clear (for the single machine setting) that it suffices to consider so-called static list scheduling policies~\cite{Pinedo}, meaning that jobs are scheduled in the order of a fixed priority list which is determined ex ante, using the given information about processing time distributions $\p_j$. We will therefore simply refer to a scheduling policy as (scheduling) algorithm. If the distributions $\p_j$ are known, \texttt{WSEPT}, that is
scheduling in order of non-increasing ratios $w_j/\E{\p_j}$ is the policy that minimizes the expectation of the total weighted completion times~\cite{Rothkopf66}.


For convenience let us write $\E{}$ instead of  $\mathbb{E}_{p\sim P}[\,\,]$, if no ambiguity arises. This paper addresses the information regime where the random variables $\p_j$ are unknown, hence also \E{\p_j} is unknown. The  only information available to the scheduler is one sample $p'_j\sim\p_j$, for all jobs $j=1,\dots,n$. In this regime, the only reasonable algorithm that exploits the given information is -arguably- to take the sample $p'_j$ as a proxy for $\E{\p_j}$.
\begin{definition}
    Algorithm \alg\ schedules jobs in non-increasing order weight over sampled processing time $w_j/p'_j$. 
\end{definition}

Our goal is to analyze this algorithm's expected cost, resp.\ expected regret. The following simple example gives a malicious input instance showing that, in general, the expected cost as well as expected regret achieved by \alg\ can be unbounded. 
\begin{example}
\label{ex:wouter_example}
Consider instance $I$ with $n-1$ jobs with deterministic processing time $\eps>0$ and job $n$ with processing time 
\[
P_n=\begin{cases}
   0 &  \text{ with probability } 1-\frac{1}{M}\,,\\
   M^2 & \text{ with probability } \frac{1}{M}\,,
\end{cases}
\]
where $M$ is large. Let $w_j=1$ for all jobs. 
\end{example}
For the following discussion let $n$ be fixed and consider $\eps\to 0$. As $\E{P_n}=M$, the algorithm that minimizes $\E{\sum_jC_j}$ is to schedule all $\eps$-jobs first, with expected cost $M$ and expected regret equal to $(1-{1}/{M})(n-1)\eps\to 0$.
As to \alg, observe that with probability $1-1/M$ the sample for job $n$ will be $p_n'=0$, which yields \alg\ to schedule job $n$ before all $\eps$-jobs. In this case, the expected cost of $\alg$ is $nM$, while $\opt{p}$ has expected cost $M$. With probability only $1/M$, the sample $p_n'=M^2$ and $\alg$ schedules job $n$ last, and just like $\opt{p}$ achieves expected cost $M$. Therefore, \alg's expected cost equals 
\begin{align*}
\cost_I(\alg)& =\mathbb{E}_{p'\sim P}\mathbb{E}_{p\sim P}\left[\sum\nolimits_jw_jC_j^{\alg(p')}(p)\right]\\& = \left(1-\frac{1}{M}\right)nM+\frac{1}{M}M=nM-(n-1)\ =\ \Theta(nM)\,,
\end{align*}
and the expected 
regret over the offline optimum $\opt{p}$ is of the same order of magnitude, as 
\begin{align*}
\text{expected regret $\alg$}&=\mathbb{E}_{p'\sim P}\mathbb{E}_{p\sim P}\left[\sum\nolimits_jw_jC_j^{\mbox{$\alg(p')$}}(p)-\opt{p}\right]\\
&=\left(1-\frac{1}{M}\right)\left[nM-M\right]\ =\ (n-1)(M-1)
\ =\ \Theta(nM)\,.
\end{align*}
Observe that the expectations are over the samples $p'\sim P$ to define \alg, and over the actual realizations of processing times $p\sim P$. If there is no ambiguity, we will omit the subscripts.

Example~\ref{ex:wouter_example} also shows more, however. If we make the assumption 
that the adversary should be required to be non-anticipatory as in~\cite{MSU99}, then the omniscient adversary $\opt{p}$ is ruled out. The fundamental question is which is the right adversary to compare with, under the given information regime. In lack of reasonable alternatives, the least to ask for is an algorithm that compares favourably to an adversary that does \emph{not} make use of the given processing time samples. In other words, we ask for an algorithm that is at least as good as the adversary $\rand$ that picks \emph{any schedule uniformly at random}. 
\begin{definition}
    Algorithm \rand\ schedules jobs in any of the $n!$ sequences uniformly at random. 
\end{definition}
We argue below that the choice of this adversary has several advantages. However, the bad news is that Example~\ref{ex:wouter_example} is built so that algorithm \alg\ has large expected regret even in comparison to the much weaker adversary $\rand$, because 
\begin{align*}
&\mathbb{E}
\left[\sum\nolimits_jw_jC_j^{\mbox{$\alg(p')$}}(p)-\rand(p)\right]
&=\left(1-\frac{1}{M}\right)\left[nM-\frac{n+1}{2}M\right]+ \frac{1}{M}\left[M-\frac{n+1}{2}M\right]\\
&=\frac12(n-1)(M-2)\ =\ \Theta(nM)\,.
\end{align*}


To formalize the question that we aim to answer, let us define a partial order on the set of algorithms as follows.
\begin{definition}
    If $\inst$ is a set of problem instances, and $\Pi$ and $\Pi'$ are two scheduling algorithms,
    write $\Pi\le_\inst \Pi'$ if $\cost_I(\Pi)\le \cost_I(\Pi')$ for every problem instance $I\in\inst$.
\end{definition}
Example~\ref{ex:wouter_example} then shows that not even  $\alg\le_\inst\rand$ holds, if $\inst$ contains all problem instances. We mainly intend to give an answer to the following question about the potential benefit of having access to samples $p'\sim\p$.
\[
    \boxed{\text{Are there conditions on the input instances $\inst$ under which  $\alg\le_\inst\rand$?}}
\]

\section{Uniform Randomization as Adversary}
\label{sec:rand}
The adversary $\rand$ seems logical for Example~\ref{ex:wouter_example}, because all job weights $w_j$ are the same. 
In general, one might define other adversaries that take into account the job weights. We argue that even in the presence of job weights the uniform random schedule $\rand$ is favourable.

The motivation lies in the fact that the expected cost achieved by $\rand$ is the average of the expected cost of the cheapest schedule and the expected cost of the most expensive schedule for any problem instance. In this relative sense, the expected cost of $\rand$ is independent of the considered problem instance. This will be made rigorous in Section \ref{sec:rog}. Comparing favourably against this instance independent benchmark shows that we have successfully leveraged the fact that we have access to samples.

Other, seemingly reasonable adversaries, e.g.\ the algorithm that schedules jobs by largest weight first, or randomized variations thereof, depend non-trivially on the problem instance. Since we do not have access to the $P_j$'s, however, showing that \alg\ outperforms such an adversary would generally not allow to distinguish between \alg\ performing well, or the adversary performing poorly. The latter is ruled out when considering \rand\ as a benchmark.

\section{Relative Optimality Gap and Better Than Random Schedules}
\label{sec:rog}

The partial order on algorithms $\leq_{\inst}$ defined in Section \ref{sec:prelim} has the shortcoming that it does not give a quantitative measure of how well an algorithm performs. In this section, we therefore define a scaled version of the usual notion for approximation ratios.

For a given problem instance $I$, denote by $L_I$ the lowest possible expected cost achieved by an optimal sequence, and by $H_I$ the highest possible expected cost (achieved by the reverse sequence). For simplicity, let us exclude for the rest of the paper instances $I$ where $L_I=H_I$, in which case trivially $\cost_I(\alg)=\cost_I(\rand)=L_I=H_I$.
One might consider to prove that \alg\ (or any other algorithm) is an $\alpha$-approximation algorithm for some $\alpha>1$, meaning $\cost_I(\alg)\le\alpha L_I$, for all  $I\in\inst$ and given $\inst$.  However, depending on the problem instance it might be that $\alpha L_I \ge H_I$, trivializing such a statement. So, if we want to assess if using the processing time samples improves upon scheduling at random, showing that an algorithm is an $\alpha$-approximation algorithm does not necessarily yield an answer for all instances. This motivates the following definition.
\begin{definition}
    The relative optimality gap $($\rog$)$ of an algorithm $\Pi$ on a problem instance $I\in \inst$ is defined as
    \begin{align*}
    \rog_I(\Pi):= 
        \frac{\cost_I(\Pi) - L_I}{H_I - L_I}\,. 
    \end{align*}
    For a set of problem instances $\inst$, define $\rog_\inst(\Pi) :=\sup_{I\in\inst}\rog_I(\Pi)$.
\end{definition}
Note that $\rog_\inst(\Pi) \le \alpha$, for some $\alpha\le 1$,  means that there is no instance $I$ where $\Pi's$ expected cost exceeds that of an optimal solution by more than ${\alpha}(H_I-L_I)$. It corresponds to a $\beta$-approximation algorithm, with $\beta=\beta(I)=1+{\alpha}(H_I/L_I-1)$. Roughly speaking, the $\rog$ is more meaningful than the notion of approximation ratios whenever $H_I-L_I$ is small, relative to $L_I$. (However this argument works vice versa, too.) One advantage of working with the $\rog$ is that it yields a meaningful statement for every single instance.  
Moreover, if $\rog_\inst(\Pi) \le \alpha$, it implies that the expected regret of $\Pi$ is at most an $\alpha$-fraction of the maximum expected regret. In that sense, the $\rog$ can be also seen as a ``regret version'' of the usual notion of approximation ratios.

As long as the set of input instances $\inst$ does include instances such as Example~\ref{ex:wouter_example}, we have that (again, considering $\eps\to 0)$
\[
\rog_\inst(\alg)\ge \frac{nM-(n-1) -M}{nM-M}\to 1\ (\text{for $M\to\infty$}\ \text{and any $n\ge 2$})\,.
\]

The following lemma defines the benchmark that we aim for in our subsequent analysis of algorithm $\alg$.
\begin{lemma}\label{lem:one_half}
    For any instance $I$ we have $\rog_I(\textup{\rand})=1/2$. Moreover, for scheduling algorithm $\Pi$ and any set of input instances $\inst$,  we have 
    $\Pi \le_\inst \rand \Leftrightarrow \rog_\inst(\Pi)\le \rog_\inst(\rand)$.
\end{lemma}
\begin{proof}
   For given instance $I=(w,\p)$ assume w.l.o.g.\ that  $w_1/\E{\p_1}\ge \cdots \ge w_n/\E{\p_n}$. The first claim follows because $L_I=\sum_{j=1}^{n}\sum_{k=j}^n w_k\E{\p_j}$, and   
   any two jobs $j,k$ have probability $1/2$ of being ordered either way by $\rand$, so by canceling the terms $w_j\E{p_j}$
   \begin{align*}
       \cost_I(\rand) - L_I&= 
			 \sum_{j=1}^{n-1} \sum_{k=j+1}^{n} \frac{1}{2} w_j \mathbb{E}(P_k) - \frac{1}{2} w_k \mathbb{E}(P_j)
			= \frac{1}{2} \left(H_I-L_I\right).
   \end{align*}The second claim then follows directly by using the definition of $\rog$, since $\rog_I(\Pi)\le \frac12$ precisely means that $\cost_I(\Pi)\le\frac12(L_I+H_I)=\cost_I(\rand)$.
\end{proof}
That means that we may proceed to find conditions on $\inst$ under which $\rog_\inst(\alg)\le 1/2$. Note that the second statement in Lemma~\ref{lem:one_half} is not true when comparing two arbitrary scheduling algorithms $\Pi$ and $\Pi'$, as $\rog_\inst(\Pi)\le\rog_\inst(\Pi')$ does not  imply that $\Pi \le_\inst \Pi'$.

For what follows it is convenient to realize that the relative optimality gap 
may be expressed in terms of the extra cost for scheduling pairs of jobs in the incorrect order. To that end, for a given instance $I=(w,\p)$, we denote this extra cost by
\[
\diff_{jk} := w_j\E{\p_k}-w_j\E{\p_j},
\]
and write 
$\prob{\Pi:j\to k}$ for the probability that $\Pi$ schedules job $j$ before $k$. 
\begin{lemma}\label{lem:delta}
Let $\Pi$ be a scheduling algorithm for instance $I=(w,\p)$, and assume w.l.o.g.\ that $w_1/\E{\p_1}\ge \cdots \ge w_n/\E{\p_n}$, then
\begin{align}\label{eq:Delta-formula-rog}
\rog_{I}(\Pi) =
\sum_{j,k: j<k} \prob{\Pi:k\to j} \frac{\diff_{jk}}{\sum_{h,\ell: h<\ell} \diff_{h\ell}}\,.
\end{align}
\end{lemma}
\begin{proof}
   The proof follows from observing that $H_I-L_I= \sum_{h,\ell: h<\ell} \diff_{h\ell}$, and since, again by canceling the terms $w_j\E{\p_j}$, 
\begin{align*}
\cost_I(\Pi) - L_I &=  \sum_{j=1}^{n-1} \sum_{k=j+1}^{n} \left(\prob{\Pi:j\to k} - 1\right)w_k \E{\p_j} + \prob{\Pi:k\to j}w_j \E{\p_k}\\
&= \sum_{j=1}^{n-1} \sum_{k=j+1}^{n} \prob{\Pi:k\to j} \diff_{jk}\,.
\end{align*}
\end{proof}
Now we can derive a simple and intuitive bound on the $\rog$ by bounding the probability of scheduling jobs in the incorrect order.
\begin{theorem}\label{thm:kappa}
    If $\Pi$ is an algorithm for instance $I$, if there exists some  $\kappa \le 1$, so that for all pairs of jobs $\prob{\Pi:j\to k} \ge \kappa$ whenever $w_j/\E{p_j} > w_k/\E{p_k}$, then 
    $
    \rog_I(\Pi) \le {1-\kappa}\,. 
    $ 
\end{theorem}
\begin{proof}
    The proof follows from Lemma~\ref{lem:delta} and $\prob{\Pi: k\to j} \le (1-\kappa)$, and because $\diff_{jk}=0$ whenever $w_j/\E{p_j} = w_k/\E{p_k}$.
\end{proof}
One way of interpreting this result, in view of the second part of Lemma~\ref{lem:one_half}, is  that $\rog\le\alpha$ means that the  expected regret is no worse than the expected regret of an algorithm that schedules each pair of jobs in the correct order with probability $1-\alpha$. Also note that the inverse statement does not hold, namely,  $\rog\le\alpha$ does not necessarily imply that all pairs of jobs are scheduled in the correct order with probability at least $1-\alpha$.


\section{Well Behaved Input Distributions}
\label{sec:results}

Here we derive our main results, namely three classes of rather natural assumptions on input distributions that allow to make use of Theorem~\ref{thm:kappa} to obtain performance guarantees for sampling based algorithm \alg. We will argue in Section~\ref{sec:final} that these results are tight in a mild sense.
An assumption that we make henceforth is that the distributions for $\p_j$ all have a density~$f_j$. The same ideas and results also work for discrete distributions, however. 

If $p'_j$,$p'_k$ are the sampled processing times, then  $\alg$ schedules $j$ before $k$ if $p'_j\le \frac{w_j}{w_k}p'_k$. By independence of processing times we have the following. 
\begin{observation}
    Consider scheduling algorithm $\alg$, then the probability for scheduling two jobs $j,k$ in order $j\to k$ is
    \begin{equation}\label{eq:formula_j_k}
    \prob{\alg: j\to k} = \int_{0}^{\infty}  f_k(y) \int_{0}^{\frac{w_j}{w_k}y} f_j(x)dxdy\,.
   \end{equation}
\end{observation}

\subsection{Symmetric Processing Time Distributions}
Intuitively speaking, if all distributions $\p_j$ are symmetric around their means $\E{\p_j}$, then for each sampled pair $p'_j,p'_k$ that gives rise to an incorrect ordering of $j$ and $k$, by symmetry, there exist samples $p_j'',p_k''$ that appear with equal probability and give rise to the correct ordering. The proof of the following theorem confirms that this intuition is essentially correct.
Let us first recall what symmetry of non-negative random variables means.
\begin{definition}
    Consider a random variable $\p$ on $\mathbb{R}$ with density $f$ and positive and finite expected value $E$, then $\p$ is symmetric if $f(E-x)=f(E+x)$ for all $x$. 
\end{definition}
Note that in our context $\p\ge 0$, which implies that $f(x)=0$ for all $x\le 0$ and all $x\ge 2E$. 
\begin{theorem}
    Consider instances $\inst$ so that all processing time distributions $\p_j$ are symmetric for all jobs $j=1,\dots,n$, then $\alg\le_\inst\rand$.
\end{theorem}
\begin{proof}
    Pick an instance $I$ with symmetric processing time distributions, and recall by Theorem~\ref{thm:kappa} and Lemma~\ref{lem:one_half} that it suffices to prove that for any pair of jobs $j,k$ with $w_j/\E{\p_j} > w_k/\E{p_k}$, we have $\prob{\alg:j\to k}\ge \frac12$, since then $\rog_\inst(\alg)\le\frac12$ by Theorem~\ref{thm:kappa}. For convenience write $E_j$ and $E_k$ for $\E{\p_j}$ and $\E{\p_k}$.
    Making use of \eqref{eq:formula_j_k}, we apply linear variable substitutions. Substituting $x$ by $x'=w_k(x-E_j)$, $y$ by $y'=w_j(y-E_k)$, and noting that $dx'=w_k\,dx$ and $dy'= w_j\,dy$, we get
    \[
     \prob{\alg: j\to k} = \frac{1}{w_jw_k}\int_{-w_jE_k}^{\infty}  f_k(E_k + \frac{y'}{w_j}) \int_{-w_kE_j}^{y'+w_jE_k- w_kE_j} f_j(E_j+\frac{x'}{w_k})dx'dy'\,.
    \]
    From symmetry of $\p_k$ it follows that $f_k(E_k + \frac{y'}{w_j}) = 0$ for $y' > w_j E_k$. Thus, we can replace the upper boundary of the outer integral by $w_j E_k$. Renaming  $x'$ back to $x$ and $y'$ to $y$, gives
    \[
     \prob{\alg: j\to k} = \frac{1}{w_jw_k}\int_{-w_jE_k}^{w_jE_k}  f_k(E_k + \frac{y}{w_j}) \int_{-w_kE_j}^{y+w_jE_k- w_kE_j} f_j(E_j+\frac{x}{w_k})dxdy\,.
    \]
    And symmetrically also
     \[
     \prob{\alg: k\to j} = \frac{1}{w_jw_k}\int_{-w_kE_j}^{w_kE_j}  f_j(E_j + \frac{x}{w_k}) \int_{-w_jE_k}^{x+w_kE_j- w_jE_k} f_k(E_k+\frac{y}{w_j})dydx\,.
    \]
    The final two expressions can be interpreted as surface integrals in $\mathbb{R}^2$. The integration domain in both cases is an isosceles right-angled triangle. These triangles are shown in Figure~\ref{fig:Qs_as_triangles_symmetric_case_flipping}. We denote the triangle corresponding to $\prob{\alg:j\to k}$ with $T_{jk}$, and analogously for $T_{kj}$. 
\begin{figure}[h]
		\centering
		\begin{tikzpicture}[scale=0.7]
			\draw[help lines, color=gray!30, dashed] (-4.9,-4.9) grid (4.9,4.9);
			\draw[->,ultra thick] (-5,0)--(5,0) node[right]{$x$};
			\draw[->,ultra thick] (0,-5)--(0,5) node[above]{$y$};
			
			\draw [ultra thick, blue] (-2,-3)
			-- (4,3)
			-- (-2,3)
			-- cycle;
			
			\draw [ultra thick, green, dashed] (-2,-3)
			-- (2,1)
			-- (2,-3)
			-- cycle;
			
			\draw [ultra thick, green, dashed] (-2,-1)
			-- (2,3)
			-- (-2,3)
			-- cycle;
			
			\draw [ultra thick, green] (1.8,-3)
			-- (1.8,-2.8)
			-- (2,-2.8);
			
			\draw [ultra thick, blue] (-2,2.8)
			-- (-1.8,2.8)
			-- (-1.8,3);
			
			\draw[->, ultra thick, green] (0.5,-2) -- (-0.5,-2);
			\draw[->, ultra thick, green] (-1,-0.5) -- (-1,0.5);
			
			\node [thick, blue] at (-1,2) {\Large $T_{jk}$};
			\node [thick, green] at (1,-2) {\Large $T_{kj}$};
			
			\node[gray, anchor = south west] at (0,3) {$w_j E_k$};
			\draw[gray,fill=gray] (0,3) circle (.5ex);
			\node[gray, anchor = north west] at (0,-3) {$-w_j E_k$};
			\draw[gray,fill=gray] (0,-3) circle (.5ex);
			\node[gray, anchor = north west] at (2,0) {$w_k E_j$};
			\draw[gray,fill=gray] (2,0) circle (.5ex);
			\node[gray, anchor = north east] at (-2,0) {$-w_k E_j$};
			\draw[gray,fill=gray] (-2,0) circle (.5ex);
		\end{tikzpicture}
		\caption{Two triangles representing the integral domains for the probabilities of \alg\ scheduling in each pairwise order. 
  The diagonal strip where the triangles do not overlap corresponds to the difference in the integral domains, proving that $\prob{\alg:j\to k} \ge \prob{\alg:k\to j}$.}
		\label{fig:Qs_as_triangles_symmetric_case_flipping}
	\end{figure}
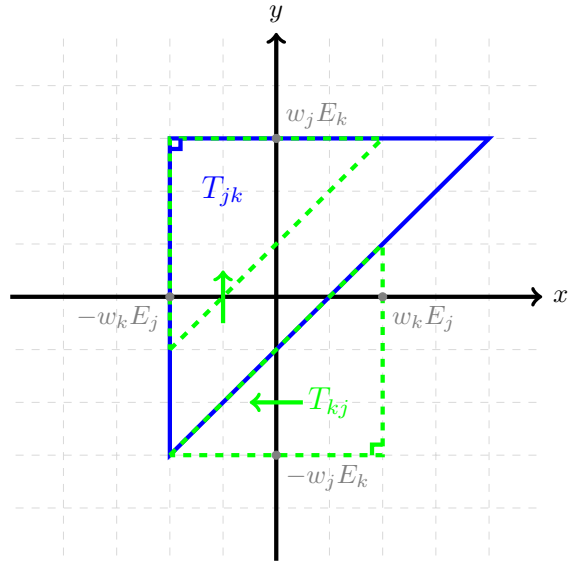
Due to the linear transformations, the symmetry of $\p_j$ and $\p_k$ translates into the fact that the integrals over these triangles remain the same after reflection through either of the coordinate axes. Mirroring the triangle $T_{kj}$  twice, once in the x-axis, and once in the y-axis, aligns its right-angle corner with the right-angle corner of the triangle $T_{jk}$. This is illustrated in Figure~\ref{fig:Qs_as_triangles_symmetric_case_flipping}. This geometric operation can also be done formally by changing the order of integration in the term for $ \prob{\alg: k\to j}$, using Fubini's theorem. It results in
\[
 \prob{\alg: k\to j} = \frac{1}{w_jw_k}\int_{w_jE_k-2w_kE_j}^{w_jE_k}  f_k(E_k + \frac{y}{w_j}) \int_{-w_kE_j}^{y- w_jE_k + w_kE_j} f_j(E_j+\frac{x}{w_k})dxdy\,.
\]
Geometrically, the assumption that $w_j E_k > w_k E_j$ yields that triangle $T_{kj}$ has shorter side-lengths than triangle $T_{jk}$. This means that the integration domains of the two otherwise identical integrals differ by the diagonal strip in Figure~\ref{fig:Qs_as_triangles_symmetric_case_flipping}, and since the integrand is non-negative, 
$\prob{\alg:j\to k}\ge \prob{\alg:k\to j}$. Therefore, as required, 
$\prob{\alg:j\to k}\ge \frac12$.
\end{proof}
Note that the proof also shows that whenever the difference between the terms $w_j/E_j$ and $w_k/E_k$ gets large, also  $\prob{\alg:j\to k}$ gets larger relative to $\prob{\alg:k\to j}$. However without knowledge about the precise processing time distributions, this does generally not yield qualitatively better bounds than $\rog\le 1/2$. We will again pick up this idea in Section~\ref{sec:far-apart}. 

\subsection{Identical Processing Time Distributions up to Scaling}
\label{sec:scaled}

Here we show that if all distributions $\p_j$ are identical up to scaling, then algorithm \alg\ performs better than random.
\begin{definition}
    Let $g(\cdot)$ be any probability density function over $\mathbb{R}_{\ge 0}$ with expectation~$1$. Define for $\lambda_j>0$, 
    $
    f_{j}(x):= \lambda_jg(\lambda_jx)\,.
    $
    Call an instance of the scheduling problem \emph{shape-uniform} if the processing time distributions $\p_j$ all have densities $f_j(\cdot)$, for the same $g(\cdot)$ and  $\lambda_j>0$, $j=1,\dots,n$.
\end{definition}
Note that $\E{\p_j}=1/\lambda_j$, $j=1,\dots,n$,  for any shape-uniform instance. Also note that exponentially distributed processing times where $f_j(x)=\lambda_je^{-\lambda_jx}$ are contained as a special case of shape-uniformity via $g(x)=e^{-x}$.

\begin{theorem}\label{thm:scaled}
    Consider instances $\inst$ so that all processing time distributions $\p_j$ are shape-uniform with the same underlying probability density function $g(\cdot)$, then 
    $\alg \le_\inst\rand$.
\end{theorem}

\begin{proof}
    Consider any two jobs $j,k$ and assume w.l.o.g.\ that $w_j/\E{\p_j}\ge w_k/\E{\p_k}$, which is equivalent to $w_j\lambda_j\ge w_k\lambda_k$. Again we make use of Theorem~\ref{thm:kappa} and Lemma~\ref{lem:one_half}, and show that $\prob{\alg:j\to k}\ge\prob{\alg:k\to j}$. 
    Making use of \eqref{eq:formula_j_k}, and substituting $x$ by $\lambda_jx$ and $y$ by $\lambda_k y$, we get
    \[
    \prob{\alg: j\to k} = \int_{0}^{\infty}  g(y) \int_{0}^{\frac{w_j\lambda_j}{w_k\lambda_k}y} g(x)dxdy\,.
    \]
    Symmetrically, we get the same term for $\prob{\alg: k\to j}$, and subtracting gives
    \[
    \prob{\alg: j\to k}-\prob{\alg: k\to j}= \int_{0}^\infty g(y) \int_{\frac{w_k \lambda_k}{w_j \lambda_j}y}^{\frac{w_j \lambda_j}{w_k \lambda_k}y} g(x) dx\,dy\,.
    \]
    Now since $w_j\lambda_j\ge w_k\lambda_k$, and since the integrand is non-negative everywhere, this last term must be non-negative, which yields the claim.
\end{proof}

\subsection{Translations of Identical Processing Time Distributions}
\label{sec:translated}
Here we show that that \alg\ performs better than random if all processing time distributions are identical up to having different expectations. This result is restricted to the special case of uniform weights $w_j=1$, $j=1,\dots,n$, however.

\begin{definition}
    Let $g(\cdot)$ be the density of some random variable with expectation $E>0$, so that $g(x)=0$ for all $x<0$. Call the processing times \emph{translation-identical} for $g(\cdot)$, if processing time $\p_j$  has density $f_j(x)=g(x-E_j+E)$, for $E_j\ge E$. 
\end{definition}
Note that by definition,  $\E{\p_j}=E_j$ for $j=1,\dots,n$.
\begin{theorem}
    Consider instances $\inst$ with uniform weights $w_j=1$ and translation-identical processing time distributions $\p_j$, for some $g(\cdot)$,  then 
    $\alg\le_\inst\rand$.
\end{theorem}

\begin{proof}
    A final time we use Theorem~\ref{thm:kappa} and Lemma~\ref{lem:one_half}, and show that $\prob{\alg:j\to k}\ge\prob{\alg:k\to j}$ for any two jobs $j,k$ with $\E{\p_j}\le\E{\p_k}$. For the sake of the proof assume w.l.o.g.\ that $f_j(x)=g(x)$ and $f_k(x)=g(x-a)$, $a=E_k-E\ge 0$, are the density functions for $\p_j$ and $\p_k$, so that $\E{\p_k}=a+\E{\p_j}=a+E$. By the same arguments as in the proof of Theorem~\ref{thm:scaled}, and using the substitution $y-a$ for $y$, we easily see that 
    \[
    \prob{\alg: j\to k}-\prob{\alg: k\to j}=\int_0^\infty g(y) \int_{y-a}^{y+a}  g(x)dx\,dy\,.
    \]
    As the integrand is non-negative, and as $a\ge 0$, the term is non-negative.
\end{proof}
The following example shows, maybe surprisingly, that the above result does not generalize to the setting with arbitrary weights. (The example uses finite discrete distributions for simplicity, but could be adapted accordingly.)
\begin{example}\label{ex:weighted_counterexample}
    Consider an instance $I$ with $n=2$ jobs, weights $w_1=1$ and $w_2=2$, and processing time distributions 
    \[
    \p_1=\begin{cases}
        1, & \text{ with probability }1-\frac{1}{M}\,,\\
        M^2 & \text{ with probability } \frac{1}{M}\,,
    \end{cases} 
    \]
    and $\p_2$ identically distributed as  $\p_1$, but independently and translated by $1+\eps$ so that $\E{\p_2}=\E{\p_1}+1+\eps$.
\end{example}
Here, again assume that $M$ is large and take limits for $\eps\to 0$. Observe that Algorithm \alg\ schedules job~1 first whenever the sampled processing time $p'_1=1$, independent of the sample for $\p_2$.  This happens
with probability $(1-\frac{1}{M})$. Scheduling job~1 first has maximal expected cost which equals $5M^2$. The optimal solution is to schedule job~2 first, with expected cost $4M^2$. One readily verifies that  
\[
\rog_I(\alg)=\frac{(1-\frac{1}{M})5M^2+\frac{1}{M}4M^2-4M^2}{5M^2-4M^2}=\frac{M^2-M}{M^2}\to 1\ (\text{for $M\to\infty$})\,.
\]


\subsection{Exponential Distributions and $\alpha$-separated Priorities}
\label{sec:far-apart}

So far we showed that using the processing time samples is at least as good as random, arguably the least one would hope for. It is not hard to come up with example instances showing that these results cannot be substantially improved without making additional assumptions on the input. For example, to drive the \rog\ to $1/2$ one can always add two jobs with large weights and almost identical input distributions, so that the gap between $\cost_I(\alg)-L_I$ and $H_I-L_I$ gets arbitrarily close to $1/2$.
Here we show how to derive a qualitatively better result, yet under rather strong assumptions on the input.

We consider the single machine scheduling problem with arbitrary weights $w_j$ and exponentially distributed processing times, so $\p_j$ has a distribution with density $f_j(x)=\lambda_je^{-\lambda_j x}$, $\lambda_j>0$. Let us define $\pi_j:=w_j\lambda_j=w_j/\E{\p_j}$ for all $j=1,\dots, n$, which can be seen as the priority of job~$j$. We know from Section~\ref{sec:scaled} that $\rog_\inst(\alg)\leq 1/2$ if $\inst$ are the instances with exponentially distributed processing times. Under an additional assumption on the priorities this can be improved.
Using \eqref{eq:formula_j_k} and the definition of the exponential distribution, elementary calculus yields the following. (The basic proof is in the appendix.)
\begin{restatable}{lemma}{lemmapjkexponential}
\label{lem:P-j-k-exponential}
    If processing times are exponentially distributed, then for any pair of jobs $j,k$
    \[
     \prob{\alg: j\to k}=\frac{\pi_j}{\pi_j+\pi_k}\,.
    \]
\end{restatable}
Then call the priorities $\pi_1,\dots,\pi_n$ of an instance with weights $w_j$ and exponentially distributed processing times \emph{$\alpha$-separated}, whenever the following is true for all pairs or jobs $j,k$: Either $\pi_j=\pi_k$, or $\pi_j$ and $\pi_k$ are at least a factor $\alpha$ apart, that is, $\max\{\pi_j/\pi_k,\pi_k/\pi_j\}\ge \alpha$. In other words, we have groups of jobs with identical priorities in each group, and the priorities across groups are at least a factor $\alpha$ apart. The intuition is that either $\pi_j=\pi_k$, in which case the order of these two jobs should not matter, or $\pi_j$ and $\pi_k$ are far apart, in which case algorithm \alg\ should have a high(er) probability for scheduling these two jobs in the correct order, hence leading to a better $\rog$ bound. This intuition is indeed correct.
\begin{theorem}\label{thm:a-separated}
    Consider instances $\inst=(w,\p)$  with exponentially distributed  processing times and $\alpha$-separated priorities $\pi_j$, for $\alpha\ge 1$,  then $\rog_\inst(\alg)\le\frac{1}{1+\alpha}$.
\end{theorem}
\begin{proof}
     To be able to use Theorem~\ref{thm:kappa} for an instance $I$, observe that it suffices to consider pairs of jobs $j,k$ with $\pi_j\ne\pi_k$, because otherwise the contribution of this pair of jobs $j,k$ to $\rog_I(\alg)$ is indeed zero in~\eqref{eq:Delta-formula-rog}.
     So take any pair of jobs $j,k$ with $\pi_j>\pi_k$, meaning that $w_j/\E{\p_j}>w_k/\E{\p_k}$ and the order $j\to k$ is optimal, then by Lemma~\ref{lem:P-j-k-exponential},
     $\prob{\alg:j\to k}=\frac{\pi_j}{\pi_j+\pi_k}=\frac{1}{1+\pi_k/\pi_j}\ge\frac{1}{1+1/\alpha}=\frac{\alpha}{\alpha+1}$, where the inequality is true because $\pi_j\ge\alpha\pi_k$. The claim now follows by using  Theorem~\ref{thm:kappa}.
\end{proof}
Admittedly, this result is still unsatisfactory in the sense that the condition of being $\alpha$-separated is not natural, and moreover, it may even seem unnecessary: Whenever a pair of jobs is $\alpha$-separated, we are good (by the above proof), and whenever it is not, we should not care too much about the order of these two jobs. Indeed this intuition can be worked out, and it gives improvements over the general bound $1/2$, in the sense that one can get a convex combination of the two bounds $1/2$ and ${1}/{(1+\alpha)}$. To substantiate this to an improved bound $< 1/2$ for \emph{all} exponential instances, however, requires additional knowledge on the input, and turns out to be both tedious and the improvement quite marginal~\cite{PuckMSc2023}. The reason is that the $\rog$ may be non-continuous at  $\diff_{jk}=0$ in \eqref{eq:Delta-formula-rog}.


\section{Conclusions}
\label{sec:final}
The three classes of processing time distributions for which we show that sample-based scheduling is better than random scheduling are quite reasonable from a practical viewpoint. More interestingly, in light of Example~\ref{ex:wouter_example}, the results can be called tight in a mild sense: First note that Example~\ref{ex:wouter_example} can clearly be tweaked from discrete to continuous distributions. The counterexample does its job because there are processing time distributions with (recall Example~\ref{ex:wouter_example} had unit weights $w_j=1$)
\begin{enumerate}[(i)]
    \item  two different shapes, one of them asymmetric,  
    \item  two different means.
\end{enumerate}
Moreover, as can be seen from Example~\ref{ex:weighted_counterexample}, as soon as we have an instance with non-uniform weights and different means, not even the \emph{same} distributions around different means allow for positive results, as long as these distributions are allowed to be asymmetric and are not shape-uniform.  One may wonder about instances with identical expected processing times, yet different distributions. But this is either not interesting because all schedules yield the same expected cost (in case of unit weights), or it would allow to replicate an example exactly analogous to Example~\ref{ex:wouter_example} by exchanging the roles of weights and expected processing times, which yields the same lower bound for instances with with arbitrary weights and identical expected processing times.

There are some open ends to take this work further. That includes the analysis for adversaries other than uniformly at random, 
the identification of other classes of distributions that would allow for positive results, 
and a less restrictive information regime with two or more samples, which would allow for algorithms that use a proxy for both expected processing time and variance. However to break Example~\ref{ex:wouter_example}, note that not even a polynomial number of samples is sufficient.
Finally, one could explore if nontrivial notions for expressing the quality of the ``learner'', i.e., the samples,  would allow to get improved performance bounds, in the spirit of learning augmented algorithms.

\bibliography{OneSample}

\appendix
\section{Proof Lemma~\ref{lem:P-j-k-exponential}}
\lemmapjkexponential*

\begin{proof}
    Using \eqref{eq:formula_j_k} and the definition of the exponential distribution, we get that $\prob{\alg:j\to k}$ equals
\begin{align*}
&\int_{0}^{\infty}  e^{-y} \int_{0}^{\frac{\lambda_j w_j}{\lambda_k w_k}y}  e^{-x}dx\,dy
\ =\ \int_{0}^{\infty} e^{-y} \left[ -e^{-x} \right]_{x=0}^{x=\frac{\lambda_j w_j}{\lambda_k w_k}y}dy\\
=&\  \int_{0}^{\infty}  e^{-y} \left[ -e^{-\frac{\lambda_j w_j}{\lambda_k w_k}y} + 1\right]dy
\ =\ \left[ \left( \frac{\lambda_j w_j}{\lambda_k w_k}+1 \right)^{-1} e^{-\left( \frac{\lambda_j w_j}{\lambda_k w_k}+1 \right)y} - e^{-y}\right]_{y=0}^{y=\infty}\\
=& \ 1 - \left( \frac{\lambda_j w_j + \lambda_k w_k}{\lambda_k w_k} \right)^{-1}
\ =\  \frac{\pi_j}{\pi_j + \pi_k}.
\end{align*}
\end{proof}
\end{document}